\documentclass[letterpaper, 10 pt, conference]{ieeeconf}  

\IEEEoverridecommandlockouts                              

\usepackage{amsmath,amssymb,amsfonts}
\usepackage{mathrsfs}
\usepackage{cases}
\usepackage{algorithm}
\usepackage{algorithmic}
\usepackage{array}
\usepackage[font=scriptsize,skip=0.1pt]{caption}

\usepackage{graphicx}
\usepackage{pgfplots}
\usepackage{tikz}
\pgfplotsset{compat=1.9}
\usepgfplotslibrary{groupplots}

\usepackage[draft]{todonotes}

\title{\LARGE \bf
Data-driven input reconstruction and experimental validation
}

\author{Jicheng Shi, Yingzhao Lian, and Colin N. Jones 
\thanks{This work has received support from the Swiss National Science Foundation under the RISK project (Risk Aware Data-Driven Demand Response), grant number 200021 175627.}
\thanks{JS, YL and CNJ are with Automatic Laboratory, Ecole Polytechnique Federale de Lausanne, Switzerland.
        {\tt\small $\{$jicheng.shi, yingzhao.lian, colin.jones$\}$@epfl.ch}}
}

\newtheorem{theorem}{Theorem}
\newtheorem{lemma}{Lemma}

\newtheorem{definition}{Definition}
\newtheorem{remark}{Remark}
\newtheorem{assumption}{Assumption}

\DeclareMathOperator{\Han}{\mathfrak{H}}
\DeclareMathOperator{\R}{\mathbb{R}}

\begin{document}

\maketitle
\thispagestyle{empty}
\pagestyle{empty}

\begin{abstract}
This paper addresses a data-driven input reconstruction problem based on Willems’ Fundamental Lemma in which unknown input estimators (UIEs) are constructed directly from historical I/O data. Given only output measurements, the inputs are estimated by the UIE, which is shown to asymptotically converge to the true input without knowing the initial conditions. Both open-loop and closed-loop UIEs are developed based on Lyapunov conditions and the Luenberger-observer-type feedback, whose convergence properties are studied. An experimental study is presented demonstrating the efficacy of the closed-loop UIE for estimating the occupancy of a building on the EPFL campus via measured carbon dioxide levels.



\end{abstract}


\section{INTRODUCTION}
Input reconstruction estimates unknown inputs based on the measured states/outputs, which finds broad applications in system supervision, sensor fault detection, and robust control~\cite{hou1998input,ansari2018input,rajamani2017observers}. This problem is of particular interest when the real-time/online measurement of inputs is not affordable or privacy-sensitive. For example, as a critical factor to indoor temperature behaviors, the number of occupants can not be measured directly by cameras or Wi-Fi due to the costs and privacy. Instead, an indirect estimation is commonly deployed based on the response of indoor CO2 level~\cite{chen2018building}. The cutting force of the machine tools is another important example, whose measurement is only feasible with a dedicated laboratory setup~\cite{zhu2012state}. 

The input reconstruction problem has been studied in a model-based setup, and various methods have been proposed based on unknown input observer (UIO)~\cite{valcher1999state,sundaram2007delayed}, optimal filters~\cite{gillijns2007unbiased}, generalized inverse approach~\cite{ansari2019deadbeat} and sliding mode observers~\cite{zhu2012state}. UIO is of special interest in our study, and most methods fall into two categories in the model-based setup. In one way, system states are measured or estimated, and are further used to reconstruct the unknown input by matrix inversion~\cite{valcher1999state}, matrix pencil decomposition~\cite{hou1998input}. In the other category of methods, states and unknown inputs are estimated concurrently, whose estimate can achieve finite step convergence~\cite{ansari2019deadbeat}.

In practice, the model of the targeted system is usually not available. Instead of running a system identification, the Willems' fundamental lemma offers a direct characterization of the system responses with an informative historical dataset~\cite{willems2005note}. This characterization provides a convenient interface to data-driven methods, and has been deployed in output prediction~\cite{markovsky2008data} and in controller design~\cite{coulson2019data, de2019formulas, markovsky2007linear, berberich2020data, lian2021adaptive,lian2021nonlinear}.

This work applies the Willems' fundamental lemma to enable the direct input reconstruction with historical data. A similar setup was studied in~\cite{turan2021data}, where the system states are assumed to be measured. This work gets rid of the state measurement and achieves unknown input reconstruction directly from output measurements. In order to stress this difference, the approach developed in this work is termed the unknown input estimator (UIE), instead of the unknown input observer (UIO). This work proposes two design schemes of data-driven UIE, whose stability is studied.

In the following, Section~\ref{sect:2} reviews the output prediction based on the Willems' Fundamental Lemma, alongside the statement of the UIE problem. The design of a stable data-driven open-loop UIE is proposed in Section~\ref{sect:3}, followed by its closed-loop counterpart in Section~\ref{sect:4}.  The proposed UIEs are validated in Section~\ref{sect:5} by simulations and an occupancy estimation experiment in a real-world building, followed by a conclusion in Section~\ref{sect:6}.

\noindent\textbf{Notations}: $\mathbf{I}_n\in\mathbb{R}^{n\times n}$ denotes an identity matrix.
Regarding a matrix $M$, its numbers of columns and rows are respectively denoted by $n_M$ and $m_M$ such that $M\in \R^{n_M \times m_M}$. Accordingly, $\text{Null(M)}$ denotes its null space. $M^g:=\{X|MXM=M\}$ is the set of generalized inverse of matrix $M$. A strictly positive definite matrix $M$ is denoted by $M\succ 0$. The dimension of a vector $s$ denoted by $n_s$. Given an ordered sequence of vector $\{s_t,s_{t+1},\dots,s_{t+L}\}$, its vectorization is denoted by $s_{t:t+L} = [s_t^\top,\dots,s_{t+L}^\top]^\top\;$.


\section{PRELIMINARIES}\label{sect:2}
A discrete-time linear time-invariant (LTI) system, dubbed $\mathcal{B}(A,B,C,D)$, is defined by
\begin{equation}\label{eqn:lin_dyn_deter}
    \begin{split}
    x_{t+1} = Ax_t + Bu_t\;,y_t = Cx_t+Du_t\;,
    \end{split}
\end{equation}
 whose states,  inputs and outputs are denoted by $x \in \R^{n_x}$, $u \in  \R^{n_u}$ and $y \in  \R^{n_y}$ respectively. The order of the system is defined as $\mathbf{n}(\mathcal{B}(A,B,C,D)):=n_x$. The lag of the system $\mathbf{l}(\mathcal{B}(A,B,C,D))$ is defined as the smallest integer $\ell$ with which its observability matrix $\mathcal{O}_{\ell}:= \left[C^\top,\left(CA\right)^\top,\dots,\left(CA^{\ell-1}\right)^\top\right]^\top$
 has rank $n_x$. An $L$-step I/O trajectory generated by this system concatenates I/O sequence by $[u_{1:L};y_{1:L}]$, and the set of all possible $L$-step trajectories generated by $\mathfrak{B}(A,B,C,D)$ is denoted by $\mathfrak{B}_L(A,B,C,D)$.

\begin{definition}
A Hankel matrix of depth $L$ constructed by a vector-valued signal sequence $s:=\{s_i\}_{i=1}^T,\;s_i\in\R^{n_s}$ is
\begin{align*}
    \Han_L(s):=
    \begin{bmatrix}
    s_1 & s_2&\dots&s_{T-L+1}\\
    s_2 & s_3&\dots&s_{T-L+2}\\
    \vdots &\vdots&&\vdots\\
    s_{L} & s_{L+1}&\dots&s_T
    \end{bmatrix}\;.
\end{align*}
\end{definition}
\vspace{1em}

Given a sequence of input-output measurements $\{u_{\textbf{d},i},y_{\textbf{d},i}\}_{i=1}^T$, the input sequence is called \textit{persistently exciting} of order $L$ if $\Han_L(u_\textbf{d})$ is full row rank. By building the following $n_c$-column stacked Hankel matrix 
\begin{align*}
    \Han_L(u_\textbf{d}, y_\textbf{d}):=\begin{bmatrix}
    \Han_{L}(u_{\textbf{d}})\\\Han_L(y_{\textbf{d}})
    \end{bmatrix}\;,
\end{align*}
we state \textbf{Willems' Fundamental Lemma} as
\begin{lemma}\label{lem:funda}\cite[Theorem 1]{willems2005note}
Consider a controllable linear system $\mathfrak{B}(A,B,C,D)$ and assume $\{u_\textbf{d}\}_{i=1}^T$ is persistently exciting of order $L+ \mathbf{n}(\mathfrak{B}(A,B,C,D))$. The condition $\text{colspan}(\Han_L(u_\textbf{d},y_\textbf{d}))=\mathfrak{B}_L(A,B,C,D)$ holds.
\end{lemma}

In the rest of the paper, the subscript $_\textbf{d}$ marks a data point from the training dataset collected offline, and $L$ is reserved for the length of the system response. 

The characterization of system response by Lemma~\ref{lem:funda} is used to develop data-driven output prediction~\cite{markovsky2008data,coulson2019data}. In~\cite{markovsky2008data}, the $N_{pred}$-step output prediction $\bar{y}_{t+1:t+N_{pred}}$ driven by an $N_{pred}$-step predicted output $u_{t+1:t+N_{pred}}$ is given by the solution to the following equations at time $t$:
\begin{subequations} \label{eqn:output_pred}
\begin{align}
    \begin{bmatrix}
        \Han_{L,init}(u_\textbf{d}) \\ \Han_{L,init}(y_\textbf{d}) \\\Han_{L,pred}(u_\textbf{d})
    \end{bmatrix} g &= 
    \begin{bmatrix}
        u_{t-N_{init}+1:t}\\
        y_{t-N_{init}+1:t}\\
        u_{t+1:t+N_{pred}}
    \end{bmatrix} \label{eqn:output_pred_1}\\
    \Han_{L,pred}(y_\textbf{d}) &=: \bar{y}_{t+1:t+N_{pred}} \label{eqn:output_pred_2}\;.
\end{align}
\end{subequations}
Two sub-Hankel matrices of output are defined by
\begin{align}
    \Han_L(y_\textbf{d}) = \begin{bmatrix}
        \Han_{L,init}(y_\textbf{d})\\\Han_{L,pred}(y_\textbf{d})
    \end{bmatrix}\;, \label{eqn:Han_sep}
\end{align}
and each of them of depth $N_{init}$ and $N_{pred}$ respectively, such that $N_{init} + N_{pred} = L$. Similarly, the Hankel matrices $\Han_{L,init}(u_\textbf{d})$ and $\Han_{L,pred}(u_\textbf{d})$ are constructed. Last but not least, the solution to~\eqref{eqn:output_pred} is well-defined if $N_{init} \geq \mathbf{l}(\mathcal{B}(A,B,C,D))$. Specifically, this condition implies that $[u_{t-N_{init}+1:t}, y_{t-N_{init}+1:t}]$ the $N_{init}$-step input output sequences preceding the current point of time can uniquely determine the underlying state $x_t$. Readers are referred to~\cite{markovsky2008data} for more details.

\subsection{Problem Statement and Inspiration}

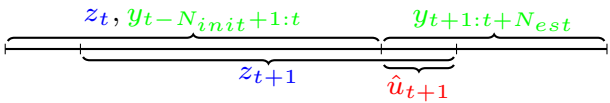
\begin{figure}[h]
  \centering

  \begin{tikzpicture}
    \draw[thick](0,0)--(8,0); 
    \draw[-] foreach~in{0,1,5,6,8}{(~,-2pt)--(~,2pt)};
  
    \draw[thick,decorate,decoration=brace](0 cm + 1pt,3pt) -- node[above,font=\scriptsize]{\scalebox{1.5}{$\color{blue} z_{t} \color{black}, \color{green} y_{t-N_{init}+1:t}$}} (5 cm - 1pt,3pt); 
    
    \draw[thick,decorate,decoration=brace](5 cm + 1pt,3pt) -- node[above,font=\scriptsize]{\scalebox{1.5}{$\color{green} y_{t+1:t+N_{est}}$}} (8 cm - 1pt,3pt);

    \draw[thick,decorate,decoration={brace, mirror}](5 cm + 1pt,-6pt) -- node[below,font=\scriptsize]{\scalebox{1.5}{$\color{red} \hat{u}_{t+1}$}} (6 cm - 1pt,-6pt);
    
    \draw[thick,decorate,decoration={brace, mirror}](1 cm + 1pt,-3pt) -- node[below,font=\scriptsize]{\scalebox{1.5}{$\color{blue} z_{t+1}$}} (6 cm - 1pt,-3pt);    
  \end{tikzpicture} 
  \caption{Diagram of input estimation at time $t+1$. Input estimate $\hat{u}_{t+1}$ can be estimated after $y_{t+N_{est}}$ is measured at time $t+N_{est}$. Differently,  output prediction $\bar{y}_{t+1}$ is computed immediately by~\eqref{eqn:output_pred}, if given $u_{t+1}$. }
  \label{fig:dig_UIE}
\end{figure}

Recall the LTI dynamics~\eqref{eqn:lin_dyn_deter}, we assume that we have an offline I/O dataset $\{u_\mathbf{d},y_\mathbf{d}\}$. During the online operation, the inputs are not measurable and thus \textbf{unknown}, and this work studies the recursive estimate of the unknown inputs from the measured outputs. In particular, the recursive estimate is generated by an  unknown input estimator (UIE), whose structure is given by
\begin{equation}\label{eqn:UIE}
    \begin{split}
    z_{t+1} &= A_{UIE}z_t+B_{UIE}d_t\;, \\
    \hat{u}_t &= [\mathbf{0}\;\mathbf{I}_{n_u}]z_t\;,
    \end{split}
\end{equation}
where $z_t := [\hat{u}_{t-N_{init}+1}^\top,\dots,\hat{u}_{t}^\top]^\top$ is vectorized $N_{init}$-step unknown input estimate, and $d_t:=y_{t-N_{init}+1:t+N_{est}}$ is the sequence of output measurements, as depicted in Figure~\ref{fig:dig_UIE}. We leave the discussion about $N_{est}$ in Section~\ref{sect:3}. Furthermore, a UIE is \textbf{stable} if $\lim\limits_{t\rightarrow \infty}\hat{u}_t - u_t \rightarrow 0$ for any initial guess $z_0$. Note that $z_t$ is the sequence of $N_{init}$-step unknown input estimate, it is reasonable to design an observable canonical form based UIE, such that the recursive estimator only update the lastest unknown input in $z_t$ (i.e. $\hat{u}_t$), and we term a UIE of this form an open-loop UIE (op-UIE). Otherwise, it is called a closed-loop UIE (cl-UIE).

The goal of this work is to design the UIE components (i.e. $A_{UIE}$ and $B_{UIE}$) directly from data $\{u_\mathbf{d},y_\mathbf{d}\}$. Inspired by the data-driven output prediction~\eqref{eqn:output_pred}, it is reasonable to formulate a similar data-driven input estimation scheme:
\begin{subequations}\label{eqn:input_est}
\begin{align}
    \begin{bmatrix}
        \Han_{L,init}(u_\textbf{d}) \\ \Han_{L,init}(y_\textbf{d}) \\\Han_{L,est}(y_\textbf{d})
    \end{bmatrix} g &= 
    \begin{bmatrix}
        u_{t-N_{init}+1:t}\\
        y_{t-N_{init}+1:t}\\
        y_{t+1:t+N_{est}}
    \end{bmatrix} \label{eqn:input_est_1}\\
    \Han_{L,est(1)}(u_\textbf{d})g &=:\bar{u}_{t+1} \label{eqn:input_est_2}\;,
\end{align}
\end{subequations}
where sub-Hankel matrices $\Han_{L,init}(u_\textbf{d})$, $\Han_{L,init}(y_\textbf{d})$, $ \Han_{L,est}(y_\textbf{d})$, $\Han_{L,est}(u_\textbf{d})$ follow a similar splitting definition in~\eqref{eqn:Han_sep} such that $N_{ini}+N_{est} = L$, and $\Han_{L,est(1)}(u_\textbf{d})$ denotes the first $n_u$ rows of $\Han_{L,est}(u_\textbf{d})$. However, this scheme~\eqref{eqn:input_est} is numerically not implementable, as input measurements in~\eqref{eqn:input_est_1} are not available. The key idea of this work is to fit this scheme~\eqref{eqn:input_est} into the general UIE structure~\eqref{eqn:UIE}.

\begin{remark} \label{rem:uyhat}
In the rest of the paper, $\bar{u}_{t+1}$ indicates the input estimate by~\eqref{eqn:input_est} given the actual previous input sequence $u_{t-N_{init}+1:t}$. $\hat{u}_{t+1}$ denotes the input estimate by~\eqref{eqn:UIE} and $z_{t+1}$. 
\end{remark}

\section{Data-driven op-UIE} \label{sect:3}
The key idea of this section is to substitute the $u_{t-N_{init}+1:t}$ in~\eqref{eqn:input_est_1} by its recursive input estimate $\hat{u}_{t-N_{init}+1:t}=:z_{t}$ in the UIE~\eqref{eqn:UIE}. For the sake of clarity, the notations in~\eqref{eqn:input_est} are simplified by
\begingroup\makeatletter\def\f@size{8.5}\check@mathfonts
\begin{align*}
    &H := \begin{bmatrix}
        \Han_{L,init}(u_\textbf{d}) \\ \Han_{L,init}(y_\textbf{d}) \\\Han_{L,est}(u_\textbf{d})
    \end{bmatrix}, 
    b := \begin{bmatrix}
        u_{t-N_{init}+1:t}\\
        y_{t-N_{init}+1:t}\\
        y_{t+1:t+N_{est}}
    \end{bmatrix},
    & H_u := \Han_{L,est(1)}(u_\textbf{d}) 
\end{align*}
\endgroup
\begin{assumption}\label{assum:PE}
The historical input signals  $\{u_\textbf{d}\}_{i=1}^T$ are persistently exciting of order $N_{init} + N_{est} + \mathbf{n}(\mathfrak{B}(A,B,C,D))$.
\end{assumption}

Under this assumption, the Hankel matrices $H$ constructed by $\{u_\mathbf{d},y_\mathbf{d}\}$ is informative enough, such that Lemma~\ref{lem:funda} guarantees that $b\in\text{colspan}(H)$. Therefore, the solution set to $Hg=b$ is non-empty, which can be characterize by
\begin{equation} \label{eqn:g_solution}
    \mathcal{T}(b):=\{g| g=Gb + \nu,\;G\in H^g, \; \nu\in\text{Null}(H)\}\;.
\end{equation}
Accordingly, $\hat{u}_{t+1}$ in~\eqref{eqn:input_est_2} is given by
\begin{align}\label{eqn:ut_op_UIE}
    \bar{u}_{t+1} = \{ H_u g|g\in\mathcal{T}(b)\}\;.
\end{align}
However,  the solution set~\eqref{eqn:g_solution} is not a singleton and therefore $\hat{u}_{t+1}$ is not necessarily unique. To ensure uniqueness, we give the following lemma
\begin{lemma}\label{lem:ut_uni}
    Let Assumption~\ref{assum:PE} holds, the set~\eqref{eqn:ut_op_UIE} is a singleton if and only if
    \begin{align} \label{eqn:null}
    \text{Null}(H) \subseteq \text{Null}(H_u)\;.
\end{align}
\end{lemma}
\begin{proof}
$(\Rightarrow)$ For any solutions $g_1,\;g_2\in\mathcal{T}(b)$ to  $Hg=b$, we have $Hg_1 - Hg_2 = H(g_1 - g_2)=0$, which indicates $(g_1 - g_2)\in\text{Null}(H)$ Therefore, by $\text{Null}(H) \subseteq \text{Null}(H_u)$, $H_u g_1 - H_u g_2 = 0$. Due the arbitrariness of $g_1$ and $g_2$, $u_{t}$ defined in~\eqref{eqn:ut_op_UIE} is a singleton. $(\Leftarrow)$ For any $G\in H^g$, $\nu \in \text{Null}(H)$, $H_u (Gb+\nu) - H_uGb  = 0$ because $\hat{u}_{t}$ by~\eqref{eqn:ut_op_UIE} is a singleton. This indicates $H_u\nu=0,\forall \nu\in\text{Null}(H)$ and therefore $\text{Null}(H) \subseteq \text{Null}(H_u)$.
\end{proof}

If the condition~\eqref{eqn:null} holds, the effect of null space $\text{Null}(H)$ in $\mathcal{T}(b)$ can be neglected. Next, by substituting the $u_{t-N_{init}+1:t}$ in~\eqref{eqn:ut_op_UIE} by $z_{t}$, we have:
\begin{align} \label{eqn:ut_op_UIE_hat}
    \hat{u}_{t+1} = H_uG[z_t^\top d_t^\top]^\top
\end{align}
, and we state the set of data-driven op-UIE candidates by
\begingroup\makeatletter\def\f@size{8.5}\check@mathfonts
\begin{equation} \label{eqn:null_AB}
    \mathcal{U}_{op}:=\left\{\begin{aligned}
        A_{UIE}\\B_{UIE}
    \end{aligned}\middle|\begin{aligned}A_{UIE} &= \begin{bmatrix}
        \mathbf{0}\;\;\;\; \mathbf{I}\\
        H_u G_u 
    \end{bmatrix}\\ B_{UIE} &= \begin{bmatrix}
                    \mathbf{0} \\
                    H_u G_y
                \end{bmatrix}\end{aligned}\;,\forall\;[G_u, G_y]=G\in H^g\right\}
\end{equation}
\endgroup
where $G_u$ and $G_y$ partitions any generalized inverse $G$, and respectively consists of $N_{init}n_u$ and $(N_{init}+N_{est})n_y$ columns. Note that an element in set $\mathcal{U}_{op}$~\eqref{eqn:null_AB} is not necessarily stable, therefore, choosing a $G\in H^g$ such that the resulting data-driven op-UIE is stable is the key ingredient in a data-driven op-UIE design procedure, which will be discussed in the next subsection~\ref{sect:desing_op_UIE}.

\begin{remark}\label{rmk:Nest}
The choices of $N_{est}$ for output measurement depends on the properties of matrices $\{B,C,D\}$ in the dynamics~\eqref{eqn:lin_dyn_deter}, which intuitively reflects how soon all the entries of inputs can affect the output. For example, if $D$ is full column rank, the effect from the input to output is instantaneous and thus $N_{est}$ can set to 1. A model-based discussion about $N_{est}$ can be found in ~\cite{valcher1999state} and~\cite{jin1997time}. The condition~\eqref{eqn:null} in Lemma~\ref{lem:ut_uni} gives a data-driven criterion of $N_{est}$ selection, which intuitively states that the variation in input will always change the output, as any $g\notin\text{Null}(H_u)$ is not in $\text{Null}(H)$, and therefore reflected as variation in the output measurements.  
\end{remark}

\subsection{Design of data-driven op-UIE}\label{sect:desing_op_UIE}
To find the stable UIE within set $\mathcal{U}_{op}$, we first characterize its stability by the following lemma
\begin{lemma}\label{lem:op_UIE_stab}
    Let Assumption~\ref{assum:PE} and condition~\eqref{eqn:null} holds, a UIE in $\mathcal{U}_{op}$ is stable if and only if $A_{UIE}$ is Schur.
\end{lemma}
\begin{proof}
Under Assumption~\ref{assum:PE} and condition~\eqref{eqn:null} holds, Lemma~\ref{lem:funda} and~\ref{lem:ut_uni} guarantee that $\bar{u}_{t+1}=u_{t+1}$ in~\eqref{eqn:ut_op_UIE}.
Therefore $\forall\; A_{UIE},\;B_{UIE}\in\mathcal{U}_{op}$, ~\eqref{eqn:ut_op_UIE} is equivalent to
\begingroup\makeatletter\def\f@size{9}\check@mathfonts
\begin{equation*}
    \begin{split}
    u_{t-N_{init}+2:t+1} 
    &= A_{UIE}u_{t-N_{init}+1:t}+B_{UIE}d_t\;,\\
    u_t &= [\mathbf{0}\;\mathbf{I}_{n_u}]u_{t-N_{init}+1:t}\;.
    \end{split}
\end{equation*}
\endgroup
Thus, we have 
 \begin{align*}
     \lim\limits_{t\rightarrow \infty}\hat{u}_t - u_t &= \lim\limits_{t\rightarrow \infty}[\textbf{0}\;\mathbf{I}_{n_u}]A_{UIE}(z_{t-1}-u_{t-N_{init}:t-1})\\
     &= \lim\limits_{t\rightarrow \infty}[\textbf{0} \;\mathbf{I}_{n_u}]A_{UIE}^t(z_0-u_{-N_{init}+1:0}). 
 \end{align*}
 The above equation converges to $0$ if and only if $A_{UIE}$ is Schur stable, and we conclude the proof.
\end{proof}

The Schur stability criterion can be validated via a semidefinite programming~\cite[Chapter 3.3]{ben2001lectures} as
\begin{align}
    \substack{A_{UIE}\\\textsc{ Schur stable}}\Longleftrightarrow \begin{cases}
         \exists\;W\succ 0\\
         \begin{bmatrix}
             W & A_{UIE}W \\
             WA_{UIE}' & W
         \end{bmatrix} \succ 0
    \end{cases}
\end{align}
hence, the design of a data-driven op-UIE is reduced to a feasibility problem:
\begin{subequations}\label{eqn:Lya_bil}
\begin{align}
    \underset{W\succ 0,A_{UIE},B_{UIE}}{\mathrm{minimize}}\;0\notag
\end{align}
subject to
\begin{align}
& A_{UIE}\;,B_{UIE}\in \mathcal{U}_{op}\label{eqn:Lya_bil_a}\\
         &\begin{bmatrix}
             W & A_{UIE}W \\
             WA_{UIE}^\top & W
         \end{bmatrix} \succ 0\label{eqn:Lya_bil_b}\;.
\end{align}
\end{subequations}
This optimization problem is NP-hard due to the bilinear matrix (BMI) inequality~\eqref{eqn:Lya_bil_b}~\cite{toker1995np}. In the rest of this section, we will tighten this BMI into a tractable linear matrix inequality (LMI)~\cite{crusius1999sufficient} and characterize the set of generalized inverse $H^g$ in $U_{op}$ via singular value decomposition (SVD).

\subsubsection{Characterization of Generalized Inverse}
Denote the SVD of matrix $H$ by $H = U\begin{bmatrix}
    S & \mathbf{0} \\
    \mathbf{0} & \mathbf{0}
\end{bmatrix}V^\top$, with $S\in\R^{n_S\times n_S}$  containing all the positive singular values. Then the generalized inverse is characterized by 
\begin{align}\label{eqn:set_Hg}
    H^g = \left\{G\middle|\begin{aligned}
        &V \left(\begin{bmatrix}
    S^{-1} & \textbf{0}\\
    \textbf{0} & \textbf{0}
\end{bmatrix}+F\right)U^\top\\
&F\in \mathbb{R}^{m_H\times n_H},\;[I_{n_S}\;\textbf{0}] F\begin{bmatrix}
    I_{n_S}\\\textbf{0}
\end{bmatrix}=\textbf{0}
\end{aligned}\right\}\;,
\end{align}
where $F$ is a any matrix of shape $H$ whose upper-left block of size $n_S\times n_S$ is zeros. For the sake of clarity, we characterize an element in $H^g$ by $G(F)$ such that 
\begin{align*}
    G(F) := V \left(\begin{bmatrix}
    S^{-1} & \textbf{0}\\
    \textbf{0} & \textbf{0}
\end{bmatrix}+F\right)U^\top
\end{align*}
Additionally, the set~\eqref{eqn:set_Hg} is indeed the set of generalized inverse as
\begin{align*}
    HG(F)H &= U\begin{bmatrix}
    S & \mathbf{0} \\
    \mathbf{0} & \mathbf{0}
\end{bmatrix} \left(\begin{bmatrix}
    S^{-1} & \textbf{0}\\
    \textbf{0} & \textbf{0}
\end{bmatrix}+F\right)\begin{bmatrix}
    S & \mathbf{0} \\
    \mathbf{0} & \mathbf{0}
\end{bmatrix}V^\top\\
& =U\begin{bmatrix}
    S & \mathbf{0} \\
    \mathbf{0} & \mathbf{0}
\end{bmatrix}V^\top = H
\end{align*}

\subsubsection{LMI Tightening} 
Before going into details, we would first intuitively explain the idea behind the design procedure. Recall the idea behind Lemma~\ref{lem:ut_uni}, we can see that the design of the $A_{UIE}$ lie in the selection of the null space of the matrix $H$ such that the set~\eqref{eqn:ut_op_UIE} is still unique, and the matrix $A_{UIE}$ is Schur stable. Hence, we only need to focus on the null space of $H$, which motivates the following LMI reformulation. Based on the characterization of $H^g$ in~\eqref{eqn:set_Hg}, any $A_{UIE}$ in our feasible set $\mathcal{U}_{op}$ is accordingly parametrized by matrix $F$ such that 
\begingroup\makeatletter\def\f@size{9}\check@mathfonts
\begin{equation} \label{eqn:AUIE_F}
\begin{aligned}
A_{UIE}(F) &= \begin{bmatrix}
        \mathbf{0}\;\;\;\; \mathbf{I}_{(N_{init}-1)n_u} \\
        H_u V(
        \begin{bmatrix}S^{-1} & \mathbf{0} \\ \mathbf{0} & \mathbf{0}
        \end{bmatrix}
        +F)U^\top \begin{bmatrix} \mathbf{I}_{N_{init}n_u} \\ \mathbf{0}\end{bmatrix}
    \end{bmatrix} \\
    &\vspace{-3mm}= N_1 + N_2FN_3, \\\vspace{-3mm}
    N_2 :&= \begin{bmatrix} \mathbf{0} \\ \mathbf{I}_{n_u} \end{bmatrix}H_u V, \;\;  N_3 = U^\top \begin{bmatrix} \mathbf{I}_{N_{init}n_u} \\ \mathbf{0}\end{bmatrix} \\\vspace{-5mm}
    N1 :&=\begin{bmatrix}
        \mathbf{0}\;\;\;\; \mathbf{I}_{(N_{init}-1)n_u} \\
        \mathbf{0}
    \end{bmatrix} + N_2 \begin{bmatrix}S^{-1} & \mathbf{0} \\ \mathbf{0} & \mathbf{0}
        \end{bmatrix} N_3\;,
\end{aligned}
\end{equation}
\endgroup
To enable the LMI reformulation, we define 
\begin{equation}\label{eqn:T12}
\begin{aligned}
     T_1 = [\mathbf{0} \;\;\mathbf{I}_{n_H-n_s}]\in\R^{(n_H-n_s)\times n_H}\;,
\end{aligned}
\end{equation}
and we denote $r = \text{rank}(T_1 N_3)$. Regarding the definition of generalized inverse and~\eqref{eqn:AUIE_F}, the operation $T_1 N_3$ select the components in $U$ related to $\text{Null}(H)$. Followed by this, we define $T_2 = [\mathbf{I}_{r} \; \mathbf{0}]E\in\R^{r\times (n_H-n_s)}$, where $E$ is the multiplication of elementary operations to execute Gauss-Jordan Elimination for $T_1 N_3$. In summary, the operation $T_2 T_1N_3$ generates the subspace of $U$ related to the $\text{Null}(H)$, and based on the aforementioned discussion, the design of $A_{UIE}$ lies within this space, which leads to the following LMI tightening.
\begin{lemma}
    The BMI constraint~\eqref{eqn:Lya_bil_b} is satisfied if $\exists\;N\in\R^{m_H\times r}$,$M\in\R^{r\times r}$ and $W\in\R^{N_{init}nu \times N_{init}nu}\succeq 0$ such that $F= NM^{-1}T_2 T_1$ and
    \begin{subequations} \label{eqn:LMI}
    \begin{numcases}{}
         & $\begin{bmatrix}
             W & \begin{matrix} N_1W + \\N_2 N T_2 T_1 N_3 \end{matrix}\\
             \begin{matrix}WN_1' + \\ (N_2 N T_2 T_1 N_3)'\end{matrix}\ & W
         \end{bmatrix} \succ 0$ \label{eqn:LMI_a} \\
         & $T_2 T_1 N_3 W = M T_2 T_1 N_3$ \label{eqn:LMI_b} 
    \end{numcases}
    \end{subequations}
\end{lemma}
\begin{proof}
    The first $n_S$ column of $F$ are zeros, because $F= NM^{-1}T_2 T_1$ and $T_1 = [\mathbf{0} \;\;\mathbf{I}_{n_H-n_S}]$. Hence, $F$ satisfies~\eqref{eqn:set_Hg} and gives a generalized inverse $G(F)$.
    
    The idea of the rest proof comes from \cite[Theorem 1]{crusius1999sufficient}. By definition of $T_2$, matrix $T_2 T_1 N_3$ is full row rank. The left-hand-sied of condition~\eqref{eqn:LMI_b} is therefore full rank as $W\succ 0$, which further ensures that $M$ is also full rank. Therefore, $M^{-1}$ exists and $T_2 T_1 N_3 = M^{-1}T_2 T_1 N_3W$. Then we get~\eqref{eqn:Lya_bil_a} from~\eqref{eqn:LMI_a} by
    \begin{align*}
        N_1W + N_2 N T_2 T_1 N_3 & =  N_1W + N_2 N M^{-1}T_2 T_1 N_3W  \\
        & \stackrel{(a)}{=}   N_1W + N_2FN_3W = A_{UIE}(F)W\;,
    \end{align*}
    where $(a)$ follows $F= NM^{-1}T_2 T_1$, and we summarize the proof.
\end{proof}

Finally, we summarize the design of data-driven op-UIE into the following LMI feasibility problem:
\begin{subequations}
\begin{align*}
    \underset{W\succ 0,M,N}{\mathrm{minimize}}\; 0
\end{align*}
subject to
\begin{align*}
        & \begin{bmatrix}
             W & A_{UIE}W \\
             WA_{UIE}' & W
         \end{bmatrix} \succ 0,\;T_2 T_1 N_3 W = M T_2 T_1 N_3\\
     & \begin{bmatrix}
             W & (WN_1 + N_2 N T_2 T_1 N_3)\\
             (WN_1 + N_2 N T_2 T_1 N_3)^\top & W
         \end{bmatrix} \succ 0 
\end{align*}
\end{subequations}
The $A_{UIE}\;,B_{UIE}\in\mathcal{U}_{op}$ is reconstructed by setting $G\in H^g$ to $G(F)$ with $F= NM^{-1}T_2 T_1$.

\section{Data-Driven cl-UIE} \label{sect:4}
Recall that an op-UIE only updates the most recent unknown input in $z_{t+1}$, i.e. $\hat{u}_{t+1}$. Similar to the concept used in Luenberger observer~\cite{luenberger1971introduction}, the key idea behind a data-driven cl-UIE is to enable the correction update of the $\hat{u}_{t-N_{init}+2:t+1}$ estimate, i.e. $z_{t+1}$, by the error between the actual measurement of $y_{t+1}$ and its data-driven predictive estimate $\hat{y}_{t+1}$.

Recall the data-driven prediction problem in Section~\ref{sect:2}, we define following matrices for the sake of clarity, 
\begin{subequations} 
\begin{align*}
    \hat{H} :=& \begin{bmatrix}
        \Han_{L,init}(u_\textbf{d})\\ \Han_{L,est(1)}(u_\textbf{d}) \\ \Han_{L,init}(y_\textbf{d})
    \end{bmatrix},\;H_y := \Han_{L,est(1)}(y_\textbf{d})\\
     \hat{b} :=&
    \begin{bmatrix}
        \hat{u}_{t-N_{init}+1:t}\\
        \hat{u}_{t+1}\\
        y_{t-N_{init}+1:t}
    \end{bmatrix}\stackrel{(a)}{=}\begin{bmatrix}
        z_t\\
        H_u(G_uz_t+G_yd_t)\\
        [I_{n_y\times N_{init}}\;\textbf{0}]d_t
    \end{bmatrix} \\
    =& \underbrace{\begin{bmatrix}
        I &\textbf{0}\\
        H_uG_u & H_uG_y\\
        \textbf{0} &[I_{n_y\times N_{init}}\;\textbf{0}]
    \end{bmatrix}}_{P(G)}\begin{bmatrix}
        z_{t}\\d_{t}
    \end{bmatrix},\;\forall\;[G_u\;G_y]=G\in H^g
\end{align*}
\end{subequations}
where $(a)$ follows~\eqref{eqn:ut_op_UIE_hat} and $P(G)$ is introduced for the sake of compactness. 
Then, similar to~\eqref{eqn:ut_op_UIE_hat}, the corresponding output prediction $\hat{y}_{t+1}$ is defined by
\begin{align}\label{eqn:y_pred_onestep}
\begin{split}
    \forall\;\hat{G}\in\hat{H}^g,\;\hat{y}_{t+1} :=&  H_y\hat{G}\hat{b}\\
             =& H_y\hat{G}P(G)\begin{bmatrix}
        z_t\\d_t
    \end{bmatrix}
\end{split}
\end{align}
 Under the Assumption~\ref{assum:PE} and condition~\eqref{eqn:null}, the Fundamental Lemma~\ref{lem:funda} and Lemma~\ref{lem:ut_uni} guarantees this equality holds for the actual output $y_{t+1}$ with respect to the actual but unknown previous inputs sequence $u_{t-N_{init}+1:t}$
 \begin{align}\label{eqn:y_onestep}
\begin{split}
    \forall\;\hat{G}\in\hat{H}^g,\;y_{t+1} 
             = H_y\hat{G}P(G)\begin{bmatrix}
        u_{t-N_{init}+1:t}\\d_t
    \end{bmatrix}
\end{split}
\end{align}
Following a Luenberger observer style design, the observer will have the following structure with $\hat{A}_{UIE},\;\hat{B}_{UIE}\in\mathcal{U}_{op}$: 
\begin{align*}
    z_{t+1} = \hat{A}_{UIE}z_t +\hat{B}_{UIE}d_t +L(y_{t+1}-\hat{y}_{t+1})\;,
\end{align*}
where $\hat{y}_{t+1}$ is given in~\eqref{eqn:y_pred_onestep} and $y_{t+1}$ is always an entry of $d_t$ as $N_{est}\geq 1$ with 
\begin{align*}
    y_{t+1} = \overbrace{[\underbrace{\textbf{0}}_{(a)} \;\mathbf{I}_{n_y} \;\underbrace{\textbf{0}}_{(b)}]}^{T_y}d_t\;,
\end{align*}
where term $(a)$ is of $N_{init}n_y$ columns and term $(b)$ is of $(N_{est}-1)n_y$ columns, and this linear mapping is denoted by $T_y$. Hence, $\forall\;\hat{A}_{UIE},\;\hat{B}_{UIE}\in\mathcal{U}_{op},\;\hat{G}\in\hat{H}^g,\;G\in H^g$, the components of an data-driven cl-UIE can be written as:
\begin{subequations}\label{eqn:cl_UIE}
\begin{align}
    A_{UIE} &= \hat{A}_{UIE} - LH_y\hat{G}P(G)\begin{bmatrix}
        \mathbf{I}_{n_y\times N_{init}}\\\textbf{0}
    \end{bmatrix}\label{eqn:cl_UIE_A}\\
    B_{UIE} &= \hat{B}_{UIE}+ T_y - LH_y\hat{G}P(G)\begin{bmatrix}
        \textbf{0}\\\mathbf{I}_{n_y\times (N_{init}+N_{est})}
    \end{bmatrix}\label{eqn:cl_UIE_B}
\end{align}
\end{subequations}
The following Theorem summarizes the stability of a data-driven cl-UIE.

\begin{theorem} [Stability of cl-UIE ] \label{thm:EFUIE}
Let Assumption~\ref{assum:PE} and condition~\eqref{lem:ut_uni} holds, for any $\hat{A}_{UIE}\;,\hat{B}_{UIE}\in\mathcal{U}_{op},\;\hat{G}\in\hat{H}^g,\;G\in H^g$, the data-driven cl-UIE in~\eqref{eqn:cl_UIE} is stable if  $\hat{A}_{UIE}-LH_y\hat{G}P(G)\begin{bmatrix}
        \mathbf{I}_{n_y\times N_{init}}\\\textbf{0}
    \end{bmatrix}$ is Schur stable.
\end{theorem}
\begin{proof}
    Similar to the proof of Lemma~\eqref{lem:op_UIE_stab}, the actual unknown input sequence satisfies the dynamics 
\begingroup\makeatletter\def\f@size{9}\check@mathfonts
\begin{equation*}
    \begin{split}
    u_{t-N_{init}+2:t+1}
    &= \hat{A}_{UIE}u_{t-N_{init}+1:t}+\hat{B}_{UIE}d_t\;,\\
    u_t &= [\mathbf{0}\;\mathbf{I}_{n_u}]u_{t-N_{init}+1:t}\;,
    \end{split}
\end{equation*}
\endgroup
Therefore, we have
\begingroup\makeatletter\def\f@size{9}\check@mathfonts
 \begin{align*}
     &\lim\limits_{t\rightarrow \infty}\hat{u}_t - u_t \\
     \stackrel{(a)}{=} &\lim\limits_{t\rightarrow \infty}[\textbf{0}\;\mathbf{I}_{n_u}]
    A_{UIE}
    (z_{t-1}-u_{t-N_{init}:t-1})\\
     =& \lim\limits_{t\rightarrow \infty}[\textbf{0} \;\mathbf{I}_{n_u}]
    A_{UIE}
    ^t(z_0-u_{-N_{init}+1:0})\;,
 \end{align*}
 \endgroup
 where $(a)$ follows the equations~\eqref{eqn:y_onestep} and~\eqref{eqn:cl_UIE}. The above equation converges to $0$ if and only if $A_{UIE}$ is Schur stable, and we conclude the proof.
\end{proof}

\begin{remark}
The design methods by Lyapunov condition~\eqref{eqn:Lya_bil}, LMI formulation~\eqref{eqn:LMI} and cl-UIE in Theorem~\eqref{eqn:cl_UIE} do not guarantee the existence of a data-driven UIE for any system. The existence problem of a UIE has been explored in a model-based setup, which shows that the existence is related to the system dynamic $\mathcal{B}(A,B,C,D)$~\cite{hou1998input,valcher1999state}. However, the existence problem within a data-driven setup is still unclear and remains for future work.
\end{remark}

\begin{remark}
In comparison with the data-driven op-UIE, we observed that the data-driven cl-UIE is more resistant to the measurement noise within the data, because it does not require any construction of the null space, which may be sensitive to the measurement noise~\cite{li1998relative}.
\end{remark}

In conclusion,  the design process and operation of op-UIE and cl-UIE are summarized in Algorithm~\ref{alg:UIE}.
\begin{algorithm} 
  \caption{Design a data-driven UIE}
  \label{alg:UIE} 
{
Given historical signals $\{u_{\textbf{d},i},y_{\textbf{d},i}\}_{i=1}^T$.
\begin{enumerate}
    \item Choose a large $N_{init}$, try $N_{est} = 1, 2,...,max_{N_{est}}$ until the condition~\eqref{eqn:null} holds.
    \item Build the UIE in the form~\eqref{eqn:UIE} by either:
    \begin{enumerate}
        \item op-UIE: Compute $G$ by either~\eqref{eqn:Lya_bil} or~\eqref{eqn:LMI}. Compute the components in~\eqref{eqn:null_AB}.
        \item cl-UIE: Choose any $G\in H^g,\;\hat{G}\in\hat{H}^g$. Design $L$ such that the stability condition in Theorem~\ref{thm:EFUIE} holds. Compute the components in~\eqref{eqn:cl_UIE}.
    \end{enumerate}
    \item From $t=0$, choose arbitrary $z_0$ and repeatedly compute~\eqref{eqn:UIE} to output $\hat{u}_t$.
\end{enumerate}
}
\end{algorithm}

\section{Simulation and experimental validation}\label{sect:5}
\subsection{Simulation}
We consider the following \textbf{unstable} LTI dynamics:
{
\setlength\extrarowheight{1pt} 
\begin{align*}
\left[
\begin{array}{@{} c|c @{}}
A & B\\
\hline
C & D\\
\end{array}
\right] = 
\left[
\begin{array}{@{} c|c @{}}
    \left[\begin{array}{@{} rrr @{}}
    0.9 & 1.4 & 0.2 \\
    0.5 & 1.5 & 1.5\\
    1.6 & 0.6 & 0.4
    \end{array}\right]   & 
    \left[\begin{array}{@{} rr @{}}
    0.5 & 1.0 \\
    0.9 & 0.3\\
    0.4 & 0.3   
    \end{array}\right] \\
\hline
    \left[\begin{array}{@{} rrr @{}}
    1.5 & 1.0& 1.4\\
    0.6 & 0.3 & 0.3
    \end{array}\right]  & 
    \gamma \left[\begin{array}{@{} rr @{}}
    2.0 & 0.8 \\
    1.4 & 1.4
    \end{array}\right]     \\
\end{array}
\right]
\end{align*}
}

Recall remark~\ref{rmk:Nest}, different $\gamma$ results $N_{est}$, and we consider $\gamma=1$ with $N_{est}=1$ and $\gamma=0$ with $N_{est}=2$. We set $N_{init} = 5$, and the historical I/O data are generated by a $50$-step trajectory excited by random inputs. Setting the initial guess to $z_0 = [0 \;0 \; \dots 0 \; 0]^\top$, the results of $\hat{u}_t(i)$ by op-UIE  are plotted in Figure~~\ref{fig:sim_all}$(a)$ and the estimation error $du_t(i) = u_t(i) - \hat{u}_t(i)$ is given in Figure~\ref{fig:sim_all}$(b)$ and $(c)$, where both proposed design schemes show fast convergence in the estimate even though the underlying dynamics is unstable.


\definecolor{myblue}{rgb}{0.20, 0.6, 0.78}
\definecolor{mygreen}{rgb}{0.2,0.8,0.2}
\definecolor{myred}{rgb}{0.5,0,0}
\definecolor{myorange}{rgb}{1,0.6,0.07}
\definecolor{mygrey}{rgb}{0.5,0.5,0.5}

 \begin{figure*}[htb!]
    \centering
    \begin{tikzpicture}
    \begin{groupplot}[
        legend columns=3,
        legend style={
    	font=\tiny},
        group style=
            {columns=3, horizontal sep=1.3cm}
        ]
    \nextgroupplot[
    title style={yshift=-3.2cm},
    title=(a),
    xmin= -5, xmax=50,
    ymin= -7, ymax= 6,
    enlargelimits=false,
    clip=true,
    grid=major,
    mark size=0.5pt,
    width=.33\linewidth,
    height=0.2\linewidth,
    ylabel = {Input $u_t(1)$},
    xlabel={Time step $t$},
    ylabel style={at={(axis description cs:-0.05,0.5)}},
    xlabel style={at={(axis description cs:0.5,-0.1)}},
    legend columns=2,
    label style={font=\tiny}, 
    ticklabel style = {font=\tiny},
    legend style={
    	font=\tiny, 
    	draw=none,
		at={(0.5,1.03)},
        anchor=south
    }    
    ]
    
    \pgfplotstableread[col sep =comma]{image/sim_D1_u1.dat}{\dat}
    
    \addplot+ [thick,mark=*,mygreen,mark options={solid, mygreen, scale=2.5}] table [x={t}, y={u1_r11}] {\dat}; 
    \addlegendentry{Actual input $u_t(1)$}
    
    \addplot+ [line width=0.5pt,mark=diamond,mark options={solid, scale=1.75},
    myred, dashed] table [x={t}, y={u1_p11}] {\dat}; 
    \addlegendentry{Estimated input $\hat{u}_t(1)$}
    
    
    \nextgroupplot[
    title style={yshift=-3.2cm},
    title=(b),
    xmin= 0, xmax=50,
    ymin= -4, ymax= 2,
    enlargelimits=false,
    clip=true,
    grid=major,
    mark size=0.5pt,
    width=.33\linewidth,
    height=0.2\linewidth,
    ylabel = {Input $u_t(1)$},
    xlabel={Time step $t$},
    ylabel style={at={(axis description cs:-0.05,0.5)}},
    xlabel style={at={(axis description cs:0.5,-0.1)}},
    legend columns=2,
    label style={font=\tiny}, 
    ticklabel style = {font=\tiny},
    legend style={
    	font=\tiny, 
    	draw=none,
		at={(0.5,1.03)},
        anchor=south
    }    
    ]
    
    \pgfplotstableread[col sep =comma]{image/sim_D1.dat}{\dat}
    
    \addplot+ [line width=1pt,mark=square,myblue,mark options={solid, myblue, scale=1.0}] table [x={t}, y={du1_LMI1}] {\dat}; 
    \addlegendentry{ $du_t(1)$ by op-UIE}

    \addplot+ [line width=0.5pt,myred,mark=o,mark options={solid,myred, scale=0.75}] table [x={t}, y={du1_Luen1}] {\dat}; 
    \addlegendentry{$du_t(1)$ by cl-UIE}
    
    \addlegendimage{thick,mark=square,myblue,mark options={solid, myblue, scale=1.0},dashed}
    \addlegendentry{ $du_t(2)$ by op-UIE}
    
    \addlegendimage{line width=0.5pt,myred,mark=o,mark options={solid,myred, scale=0.75},
    myred, dashed}
    \addlegendentry{$du_t(2)$ by cl-UIE}
    
    \nextgroupplot[
    title style={yshift=-3.2cm},
    title=(c),
    xmin= 0, xmax=50,
    ymin= -1.0, ymax= 0.5,
    enlargelimits=false,
    clip=true,
    grid=major,
    mark size=0.5pt,
    width=.33\linewidth,
    height=0.2\linewidth,
    ylabel = {Input $u_t(1)$},
    xlabel={Time step $t$},
    ylabel style={at={(axis description cs:-0.1,0.5)}},
    xlabel style={at={(axis description cs:0.5,-0.1)}},
    legend columns=2,
    label style={font=\tiny}, 
    ticklabel style = {font=\tiny},
    legend style={
    	font=\tiny, 
    	draw=none,
		at={(0.5,1.03)},
        anchor=south
    }    
    ]
    
    \pgfplotstableread[col sep =comma]{image/sim_D0.dat}{\dat}
    
    \addplot+ [line width=1pt,mark=square,myblue,mark options={solid, myblue, scale=1.0}] table [x={t}, y={du2_LMI1}] {\dat}; 
    \addlegendentry{ $du_t(1)$ by op-UIE}

    \addplot+ [line width=0.5pt,myred,mark=o,mark options={solid,myred, scale=0.75}] table [x={t}, y={du2_Luen1}] {\dat}; 
    \addlegendentry{$du_t(1)$ by cl-UIE}
    
    \addlegendimage{thick,mark=square,myblue,mark options={solid, myblue, scale=1.0},dashed}
    \addlegendentry{ $du_t(1)$ by op-UIE}
    
    \addlegendimage{line width=0.5pt,myred,mark=o,mark options={solid,myred, scale=0.75},
    myred, dashed}
    \addlegendentry{$du_t(2)$ by cl-UIE}

    \end{groupplot}
    
    \begin{groupplot}[
        legend columns=3,
        legend style={
    	font=\tiny},
        group style=
            {columns=3, horizontal sep=1.3cm}]
    \nextgroupplot[        axis y line*=right,
    xmin= -5, xmax=50,
    ymin= -7, ymax= 6,
    enlargelimits=false,
    clip=true,
    mark size=0.5pt,
    width=.33\linewidth,
    height=0.2\linewidth,
    axis x line=none,
    axis y line=none,
    legend columns=2,
    label style={font=\tiny}, 
    ticklabel style = {font=\tiny},
    legend style={
    	font=\tiny,
    	draw=none,
		at={(0.5,1.03)},
        anchor=south
    }  ]

   \nextgroupplot[
    axis y line*=right,
    ymin= -2, ymax= 4,
    xmin= 0, xmax=50,
    enlargelimits=false,
    ylabel = {Input $u_t(2)$},
    ylabel style={at={(axis description cs:1.05,0.5)}},    
    clip=true,
    mark size=0.5pt,
    width=.33\linewidth,
    height=0.2\linewidth,
    axis x line=none,
    legend columns=2,
    label style={font=\tiny}, 
    ticklabel style = {font=\tiny},
    legend style={
    	font=\tiny,
    	draw=none,
		at={(0.5,1.03)},
        anchor=south
    }    
    ]
    
    \pgfplotstableread[col sep =comma]{image/sim_D1.dat}{\dat}
    
    \addplot+ [thick,mark=square,myblue,mark options={solid, myblue, scale=1.0},dashed] table [x={t}, y={du1_LMI2}] {\dat}; 
    
    \addplot+ [line width=0.5pt,myred,mark=o,mark options={solid,myred, scale=0.75},
    myred, dashed] table [x={t}, y={du1_Luen2}] {\dat}; 

   \nextgroupplot[
    axis y line*=right,
    xmin= 0, xmax=50,
    ymin= -2, ymax= 4,
    enlargelimits=false,
    ylabel = {Input $u_t(2)$},
    ylabel style={at={(axis description cs:1.05,0.5)}},    
    clip=true,
    mark size=0.5pt,
    width=.33\linewidth,
    height=0.2\linewidth,
    axis x line=none,
    legend columns=2,
    label style={font=\tiny}, 
    ticklabel style = {font=\tiny},
    legend style={
    	font=\tiny,
    	draw=none,
		at={(0.5,1.03)},
        anchor=south
    }    
    ]
    
    \pgfplotstableread[col sep =comma]{image/sim_D0.dat}{\dat}
    
    \addplot+ [thick,mark=square,myblue,mark options={solid, myblue, scale=1.0},dashed] table [x={t}, y={du2_LMI2}] {\dat}; 
    
    \addplot+ [line width=0.5pt,myred,mark=o,mark options={solid,myred, scale=0.75},
    myred, dashed] table [x={t}, y={du2_Luen2}] {\dat}; 

    \end{groupplot}
    \end{tikzpicture}
    \caption{Simulation: input estimation by op-UIE and cl-UIE. (a): op-UIE, $\gamma=1$, input $u_t(1)$. (b): two UIEs: $\gamma=1$. (c): two UIEs: $\gamma=0$.} 
    \label{fig:sim_all}  
 \end{figure*}
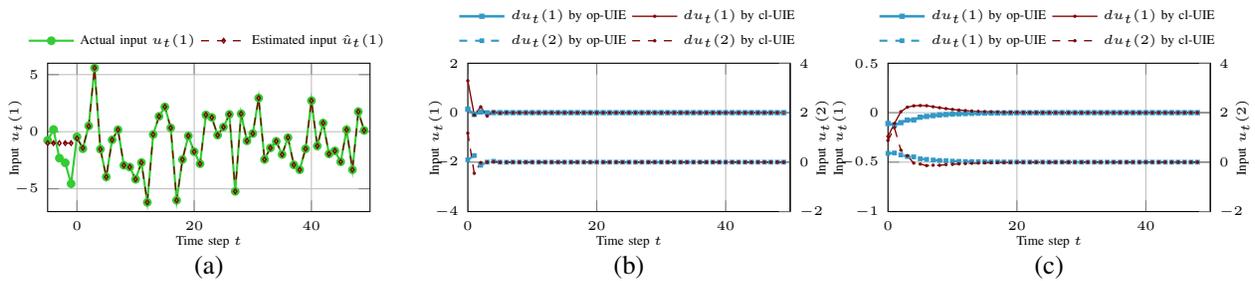

\subsection{Experiment}
This experiment is carried out on a whole building, named Polydome, on the EPFL campus, and we estimate the number of occupants by the indoor CO2 level measurement. Although the building dynamics is nonlinear due to the ventilation system, it has a good linear approximation when the ventilation flow rate is constant~\cite{cali2015co2}. Under the assumption that the CO2 generation rate per person doing office work is relatively constant, the proposed schemes in this work are feasible. The offline dataset contains indoor temperature, weather condition, heat pump power, CO2 level, and occupant number recorded by manual headcount (i.e., online measurement is not affordable). The indoor CO2 level is measured as the averaged value from four air quality sensors, whose installation locations are shown in Figure~\ref{fig:sensor}. Data from five weekdays are used to build the Hankel matrix, and the proposed data-driven cl-UIE\footnote{The op-UIE does not give good performance in this experiment due to the measurement noise within the data and the nonlinearity of the underlying dynamics.} is compared with linear regression (LR) and Gaussian process regression (GR). Note that the building is empty outside the office hour, i.e., between $7:00$ PM and $7:00$ AM, we enforce $\hat{u}_t = 0$ within this time interval to improve the estimate. The results are plotted in Figure~~\ref{fig:exp}, from which one can see that the proposed cl-UIE scheme is better than LR and slightly worse than GR in terms of mean absolute error (MAE). However, the UIE better tracks the occupancy trajectory while GR shows significant fluctuations in its estimates. 



\begin{figure}[htbp]
    \centering
    \includegraphics[width=0.2\textwidth]{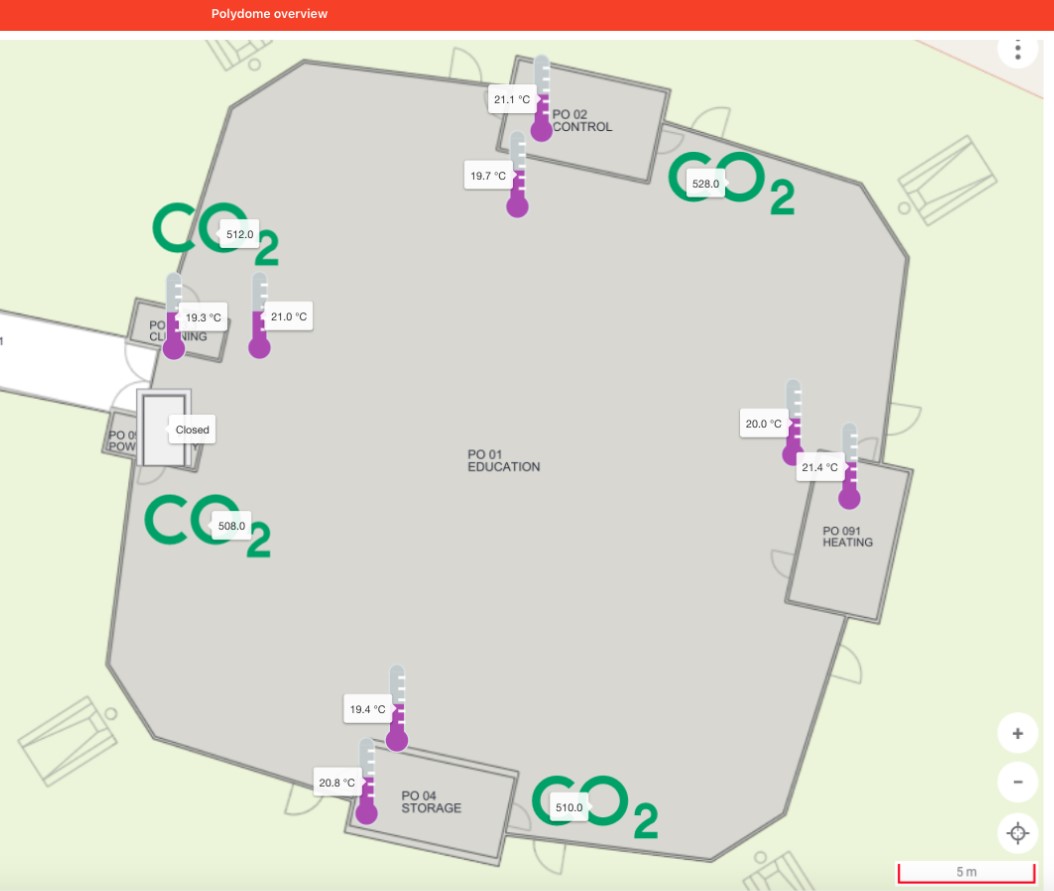}
    \caption{{\scriptsize Position of Air quality sensors in the Polydome}}
    \label{fig:sensor}
\end{figure}



\definecolor{myblue}{rgb}{0.20, 0.6, 0.78}
\definecolor{mygreen}{rgb}{0.2,0.8,0.2}
\definecolor{myred}{rgb}{0.5,0,0}
\definecolor{myorange}{rgb}{1,0.6,0.07}
\definecolor{mygrey}{rgb}{0.5,0.5,0.5}

 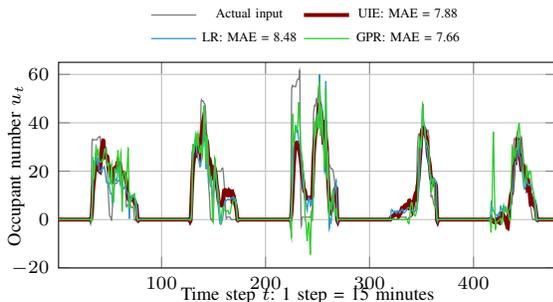
\begin{figure}[htbp]
    \centering
    \begin{tikzpicture}
    \begin{axis}[
    xmin= 1, xmax=480,
    ymin= -20, ymax= 65,
    enlargelimits=false,
    clip=true,
    grid=major,
    mark size=0.5pt,
    width=0.95\linewidth,
    height=0.5\linewidth,
    ylabel = {Occupant number $u_t$},
    xlabel={Time step $t$: 1 step = 15 minutes},
    ylabel style={at={(axis description cs:-0.05,0.5)}},
    xlabel style={at={(axis description cs:0.5,-0.05)}},
    legend columns=2,
    label style={font=\scriptsize}, 
    ticklabel style = {font=\scriptsize},
    legend style={
    	font=\tiny,
    	draw=none,
		at={(0.5,1.03)},
        anchor=south
    }    
    ]
    
    \pgfplotstableread[col sep =comma]{image/exp.dat}{\dat}
    \addplot+ [line width=0.5pt, mark=, mygrey,smooth] table [x={t}, y={u}] {\dat}; 
    \addlegendentry{Actual input}
    \addplot+ [line width=1.5pt, mark=, myred,smooth] table [x={t}, y={u_Luen}] {\dat};
    \addlegendentry{UIE: MAE = 7.88}
    \addplot+ [line width=0.5pt, mark=, myblue,smooth] table [x={t}, y={u_LR}] {\dat};
    \addlegendentry{LR: MAE = 8.48}    
    
    \addplot+ [line width=0.5pt, mark=, mygreen,smooth] table [x={t}, y={u_GR}] {\dat};
    \addlegendentry{GPR: MAE = 7.66}
    

    \end{axis}
    \end{tikzpicture}
    \caption{Comparison of occupant number estimation by the data-driven UIE, LR and GPR. Mean absolute error (MAE) is computed for the data during the work time.}
    \label{fig:exp}  
 \end{figure}



\section{Conclusions}\label{sect:6}
This work proposes two data-driven UIE design schemes based on the Lyapunov condition and the Luenberger-observer-type feedback. The stability of the proposed schemes is discussed, and their efficacy is validated by numerical simulations and a real-world experiment of occupancy estimation.





\bibliographystyle{./bib/IEEEtran}
\bibliography{ref.bib}

\end{document}